\documentclass[draftcls, onecolumn]{IEEEtran}

\usepackage{hyperref}
\usepackage{tabularx}
\usepackage{amssymb}

\hypersetup{citecolor=blue,pdfborder=0 0 0,pdfstartview={FitH}}
\ifCLASSINFOpdf
 \usepackage[pdftex]{graphicx}
\else
\fi
\usepackage{subcaption}

\usepackage{amsmath}

\usepackage{color}

\usepackage{comment}

\usepackage{soul}

\newtheorem{theorem}{Theorem}
\newtheorem{lemma}[theorem]{Lemma}
\newtheorem{algorithm}{Algorithm}
\newtheorem{definition}{Definition}
\newtheorem{problem}{Problem}
\newtheorem{example}{Example}

\newenvironment{proof}[1][Proof]{\begin{trivlist}
\item[\hskip \labelsep {\bfseries #1}]}{\end{trivlist}}

\newcommand{\qed}{\nobreak \ifvmode \relax \else
      \ifdim\lastskip<1.5em \hskip-\lastskip
      \hskip1.5em plus0em minus0.5em \fi \nobreak
      \vrule height0.75em width0.5em depth0.25em\fi}

\begin{document}

%
\title{Coding for Improved Throughput Performance in Network Switches}

\author{Rami~Cohen,~\IEEEmembership{Graduate Student Member,~IEEE,}
        and~Yuval~Cassuto,~\IEEEmembership{Senior Member,~IEEE}
\thanks{The authors are with the Department of Electrical Engineering,
Technion - Israel Institute of Technology, Haifa, Israel (email: rc@campus.technion.ac.il, ycassuto@ee.technion.ac.il)}
\thanks{Parts of this work were presented at the 2015 IEEE International Symposium on Information Theory (ISIT), Hong Kong, China.}
\thanks{This work was supported in part by the Israel Science Foundation, by the Intel ICRI-CI center, and by the Israel Ministry of Science and Technology.}
}


\maketitle
\maketitle

\begin{abstract} 
Network switches and routers need to serve packet writes and reads at rates that challenge the most advanced memory technologies. As a result, scaling the switching rates is commonly done by parallelizing the packet I/Os using multiple memory units. For improved read rates, packets can
be coded upon write, thus giving more flexibility at read time to achieve higher utilization of the memory units. This paper presents a detailed study of coded network switches, and in particular how to design them to maximize the throughput advantages over standard uncoded switches. Toward that objective the paper contributes a variety of algorithmic and analytical tools to improve and evaluate the throughput performance. The most interesting finding of this study is that the placement of packets in the switch memory is the key to both high performance and algorithmic efficiency. One particular placement policy we call "design placement" is shown to enjoy the best combination of throughput performance and implementation feasibility.

\end{abstract}


%
\IEEEpeerreviewmaketitle

\section{Introduction}
%


With the increasing demand for network bandwidth, network switches (and routers) face the challenge of serving growing data rates. Currently the most viable way to scale switching rates is by parallelizing the writing and reading of packets between multiple memory units (MUs) in the switch fabric. However, this introduces the problem of {\em memory contention}, whereby multiple requested packets need to access the same bandwidth-limited MUs. Our ability to avoid such contention in the write stage is limited, as the reading schedule of packets is not known upon arrival of the packets to the switch. Thus, efficient packet placement and read policies are required, such that memory contention is mitigated.

For greater flexibility in the read process, coded switches introduce \textit{redundancy} to the packet-write path. This is done by calculating additional coded chunks from an incoming packet, and writing them along with the original packet chunks to MUs in the switch memory. A coding scheme takes an input of $k$ packet chunks and encodes them into a codeword of $n$ chunks ($k \le n$), where the redundant $n-k$ chunks are aimed at providing improved read flexibility. Thanks to the redundancy, only a subset of the coded chunks is required for reconstructing the original (uncoded) packet. Thus, packets may be read even when only a part of their chunks is available to read without contention. One natural coding approach is to use $[n, k]$ \textit{maximum distance separable} (MDS) codes, which have the attractive property that \textit{any} $k$ chunks taken from the $n$ code chunks can be used for the recovery of the original $k$ packet chunks. Although MDS codes provide the maximum flexibility, we show in our results that good switching performance can be obtained even with much weaker (and lower cost) codes, such as binary cyclic codes.

In the coded switching paradigm we propose in this paper, our objective is to maximize the number of full packets read from the switch memory simultaneously in a read cycle. The packets to read at each read cycle are specified in a request issued by the control plane of the switch. As we shall see, coding the packets upon their write can significantly increase the number of read packets, in return to a small increase in the write load to store the redundancy. Thus coding can significantly increase the overall switching throughput. In this paper we identify and study two key components for high-throughput coded switches: 1) {\em Read algorithms} that can recover the maximal number of packets given an arbitrary request for previously written packets, and 2) {\em Placement policies} determining how coded chunks are placed in the switch MUs. Our results contribute art and insight for each of these two components, and more importantly, they reveal the tight relations between them. At a high level, the choice of placement policy can improve both the {\em performance} and the {\em computational efficiency} of the read algorithm. To show the former, we derive a collection of analysis tools to calculate and/or bound the performance of a read algorithm given the placement policy in use. For the latter, we show a huge gap between an NP-hard optimal read problem for one policy ({\em uniform} placement), and extremely efficient optimal read algorithms for two others ({\em cyclic} and {\em design} placements).

The use of coding for improved memory read rates joins a large body of recent work aimed at objectives of a similar flavor, see e.g. the survey in \cite{Dimakis}. In \cite{Joshi, Joshi2}, the effect of MDS coding on content download time was analyzed for two content access models, where an improvement in performance was achieved. In \cite{Liang}, latency delay was reduced by choosing MDS codes of appropriate rates. Latency comparison between a simple replication scheme and MDS codes was pioneered by Huang et al. \cite{Huang} using queuing theory. It was shown that for $k=2$, the average latency for serving a packet decreases significantly when a certain scheduling model is used. This analysis was later extended by Shah et al. in \cite{Shah, Shah2}, where bounds on latency performance under multiple scheduling policies were investigated. In~\cite{switch_codesISIT13} and then in~\cite{switch_codesISIT15}, switch coding is done under a strong model guaranteeing simultaneous reconstruction of worst-case packet requests.

This paper is structured as follows. In Section \ref{sec:problem_formulation}, we provide the {\em switch setting} and define formally the problem of maximizing the read throughput. We choose a simple model of a shared-memory switch, which allows defining a clean {\em throughput-optimization} problem. It is important to note that all the paper's results demonstrated on this simple model can be extended to more realistic setups, with the same underlying ideas at play.  In Section \ref{sec:placement_policies}, we define a {\em full-throughput} instance as one in which the switch is able to read all the requested packets in the same read cycle. Full-throughput instances are the most desired operation mode for a switch, because there is no need for queueing unfulfilled packet requests. We derive necessary and sufficient conditions for an instance to be full-throughput, and specify placement policies motivated by these conditions. Read algorithms are provided in Section \ref{sec:read_algorithms} for maximizing the instantaneous throughput at a read cycle. For the cyclic and design placements we show efficient polynomial time optimal read algorithms, which are also practical enough to implement in a switch environment. Probabilistic analysis of the average read throughput is provided in Section \ref{sec:prob_analysis}. We derive an upper bound for the uniform placement, a lower and an upper bound for the cyclic placement, and exact full-throughput analysis for the design placement. Simulations results are given in Section \ref{sec:simulations}. Finally, the paper is concluded in Section \ref{sec:conclusions}.

\section{Problem Setting and Formulation}
\label{sec:problem_formulation}

Given our objective to improve the switch throughput, we now define the system setting for our proposed solution, and pose the throughput-maximization problem within this setting.

\subsection{The switch setting}
Consider a switch composed of $N$ parallel memory units (MUs) serving writes and reads of incoming and outgoing packets, respectively. Each MU is capable of storing $B$ bits on a write cycle, and retrieving $B$ bits on a read cycle. A data packet is of fixed size $W>B$ bits, that is, too large to fit in a single MU. Assuming for simplicity that $W/B$ is an integer, an incoming packet is partitioned into $k \buildrel \Delta \over = W/B$ chunks, each stored in a distinct MU on the same write cycle. Upon read request of a packet, all $k$ chunks of the packet need to be retrieved by the $k$ MUs storing it, after which it can be delivered to the output port. Because there are multiple packet requests pending on the same MUs simultaneously, contention may occur between chunks of different packets stored in the same MUs.

To reduce the amount of contention in packet reading, we propose in this paper to encode the incoming packets with an $[n,k]$ code, which means that the $k$ chunks of the data packet are encoded to $n \ge k$ chunks. The $n$ encoded chunks are of size $B$ bits each, and they are stored in $n$ distinct MUs out of the $N$ MUs in the system ($1 \le k \le n \le N$). Between packets overlap is allowed, i.e., chunks of two or more packets may share one or more MUs. For the code we mostly\footnote{Part of our results in the sequel do not require the code to have such a strong property.} assume the maximum distance separable (MDS) property, which means that any subset of $k$ chunks of the $n$ encoded chunks is sufficient for recovering the original packet. We mention here the Reed-Solomon (RS) codes \cite{Moon, Lin}, which are an important family of MDS codes widely used in storage systems for improved reliability. An RS code exists for every choice of $k \le n \le q$, where $q$ is the code alphabet size, which is a prime power. RS encoding/decoding can be performed efficiently \cite{Moon, Jheng}.

In a typical switch, a large number of packets is stored in memory at any given time. Out of these many packets, $L$ particular packets are requested at each read cycle. To maximize the read throughput, at each read cycle the switch needs to recover a maximal number of the $L$ requested packets. We next define this throughput maximization problem formally.

\subsection{The maximal-throughput read problem}\label{sec:read_problem}
A request arrives for $L$ packets, with the objective to read as many out of these packets in a single read cycle. The locations of each packet's chunks are known, and we wish to find methods for reading as many packets as possible simultaneously, with the constraint that each MU can be accessed \textit{only once} in a read cycle, delivering at most one size-$B$ chunk. An instance of the problem is illustrated in Fig. \ref{fig:nkMTP_illus}, where encoded data chunks of multiple packets appear in the same column representing an MU.
Let us denote by $L^{*}$ the maximal number of packets that can be read, out of the $L$ packets requested from the switch memory. We consider the following notion of throughput as a performance measure.
\begin{definition}(Instantaneous Throughput)
The instantaneous throughput $\rho$ of the system is defined as
\begin{equation}
\rho = \frac{{L^{*}k}}{N}.
\end{equation}
\end{definition}
That is, $\rho$ is the fraction of active MUs serving packets out of the $N$ MUs in the system, and it is a monotonically increasing function of $L^*$. Clearly $0\leq \rho\leq 1$, because the total number of read chunks cannot be more than $N$. Note that given $L$, maximizing the instantaneous throughput is equivalent to maximizing $L^{*}$, because $k$ and $N$ are constants. In the sequel we refer to the instantaneous throughput as simply {\em throughput}. Later in the paper we also discuss the {\em average throughput} $\bar{\rho}$, defined as the value of $\rho$ averaged over read cycles.

We name the problem of maximizing the throughput $\rho$ as the \textit{$[n,k]$-maximal throughput problem}, or nkMTP. Recall that for reading a packet, $k$ MUs are required, which are not used to read chunks of any other packet. Thus, an nkMTP solution amounts to finding the maximal number of \textit{disjoint} $k$-sets, leading to the following set-theory formulation of nkMTP. Consider the $N$ MUs as the elements of the set $\mathcal{S} = \left\{ {0,1,2,...,N-1} \right\}$. Each packet $i=1,2,...,L$ is stored in MUs indexed by a subset $\mathcal{S}_i$ of $\mathcal{S}$, where $\left| {{\mathcal{S}_i}} \right| = n$ and the subsets may overlap. Then nkMTP can be formulated as follows.
\begin{problem}{(nkMTP)}
\label{nkMTP_set}

\textbf{Input}: The set $\mathcal{S} = \left\{ {0,1,2,...,N-1} \right\}$ and $L$ subsets of $\mathcal{S}$, $\mathcal{S}_i \subseteq \mathcal{S}$, such that $\left| {{\mathcal{S}_i}} \right| = n$.

\textbf{Output}: Subsets $\mathcal{S}'_i \subseteq \mathcal{S}_i$ such that $\left| \mathcal{S}'_i \right| = k$, $\mathcal{S}'_i \cap \mathcal{S}'_j = \emptyset$ ($i \ne j$) and the number of subsets is maximal.
\end{problem}

\begin{example}
\label{ex1}
Consider an nkMTP instance with $N=5, L=3$ and $n=3$, where the packets are stored in the MUs indexed by the sets ${\mathcal{S}_1} = \left\{ {0,1,2} \right\}, {\mathcal{S}_2} = \left\{ {1,3,4} \right\}, {\mathcal{S}_3} = \left\{ {2,3,4} \right\}$. If $k=n=3$, at most one packet can be read, since ${\mathcal{S}_i} \cap {\mathcal{S}_j} \ne \emptyset$ for $i,j \in \left\{ {1,2,3} \right\}$. If $k=2$, a possible solution is $\mathcal{S}'_1 = \left\{ {0,1} \right\}$ and $\mathcal{S}'_2 = \left\{ {3,4} \right\}$ with $L^*=2$. Finally, if $k=1$ all the packets can be read, and one possible solution is ${{\cal S}'_1} = \left\{ 0 \right\},{{\cal S}'_2} = \left\{ 1 \right\}$ and $\mathcal{S}'_3 = \left\{ 2 \right\}$.
\end{example}

\begin{figure}[t]
\centering
\includegraphics[scale=0.65]{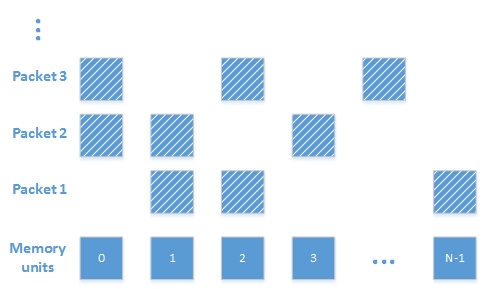}
\caption{Illustration of nkMTP. The patterned squares represent encoded data chunks ($n=3$).}
\label{fig:nkMTP_illus}
\end{figure}

An nkMTP instance can be represented as a graph as well. Consider a \textit{bipartite} graph $G = \left( {{X_G},Y_G, {E_G}} \right)$, where $X_G$ and $Y_G$ are the two disjoint sets of vertices of $G$, and $E_G$ is the set of edges of $G$. Thinking of $X_G$ as packets and of $Y_G$ as MUs, a vertex $x \in X_G$ is connected to a vertex $y \in Y_G$ if one of the encoded chunks of the packet $x$ is stored in the MU $y$. In Fig. \ref{fig_ex1}, Example $\ref{ex1}$ is represented on a graph. This graph interpretation will be used later to obtain further insights on the problem. An algorithm that guarantees maximal throughput for any instance is called an {\em optimal read algorithm}. A straightforward approach for solving an nkMTP instance is to consider all possible assignment configurations of MUs to packets. However, this approach is clearly inefficient as its complexity scales exponentially with $L$. We observe that polynomial-time optimal read algorithms exist in the general case for specific values of $k$ and $n$. These read algorithms are obtained by interpreting nkMTP for these parameters as known \textit{graph matching} problems whose efficient solutions are known.
\begin{theorem}
For $k=1, n \ge 1$ or $k=n=2$, nkMTP is solvable in polynomial time.
\label{th:k1poly}
\end{theorem}
\begin{proof}
Consider a graph representation $G$ of an nkMTP instance. When $k=1, n \ge 1$, maximizing the throughput is equivalent to finding a \textit{maximum bipartite matching} \cite{Bondy} in $G$. That is, a subgraph of $G$ with the largest number of \textit{matched} pairs $(x,y)$, $x \in X_{G}, y \in Y_{G}$, such that each pair is connected by an edge and the edges are pairwise non-adjacent. When $k=n=2$, consider the $N$ MUs as the vertices of a (uni-partite) graph, where an edge in this graph connects two MUs shared by the same packet. A maximum matching in this graph will provide the largest number of disjoint pairs of MUs, each pair serving a packet, corresponding to a maximum-throughput solution. Efficient maximum-matching algorithms are known in both cases \cite{Bondy}. \qed
\end{proof}

\begin{figure}[t]
\centering
\includegraphics[scale=0.74]{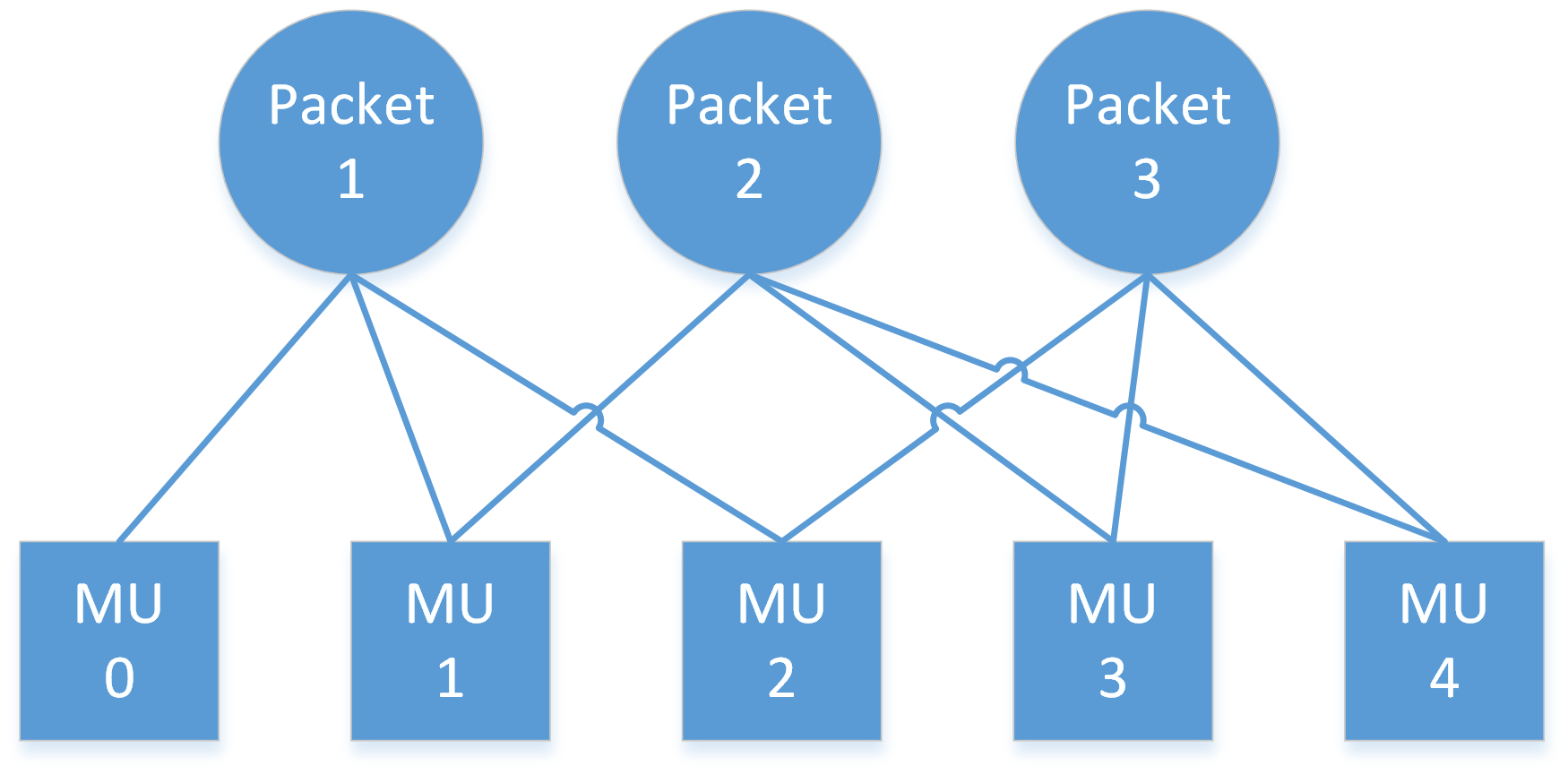}
\caption{nkMTP from Example \ref{ex1} formulated on a graph. There are three packets, each stored as $n=3$ encoded chunks in $n$ MUs.}
\label{fig_ex1}
\end{figure}

In practice, larger $k$ and $n$ values might be of interest. However, nkMTP turns out to be NP-hard in this case, as shown in the following theorem.

\begin{theorem}
nkMTP is NP-hard for $3 \le k \le n$.
\label{np_hard}
\end{theorem}
To prove Theorem \ref{np_hard}, we \textit{reduce} the $l$-set packing ($l$-SP) problem \cite{Hazan}, known to be NP-hard, to nkMTP. In $l$-SP, there are $L$ sets, each of size $l$, and the problem is to find the maximal number of pairwise disjoint sets. By the reduction, we basically show that an $l$-SP instance can be transformed to an nkMTP instance with any $k$ and $n$ in the range, which implies that nkMTP is at least as hard as $l$-SP. The details of the reduction are provided in Appendix \ref{np_hard_detailed_proof}. The consequence of the hardness result of Theorem \ref{np_hard} is that no efficient optimal algorithms are expected to be found for solving (an arbitrary instance of) nkMTP when $3 \le k \le n$.

Surprisingly, this hardness result does not imply the intractability of optimal coded switching. The main observation we make in this work is that {\em clever chunk placement} at the write path can yield more structured nkMTP instances, which do admit efficient optimal read algorithms. In the rest of this paper we develop algorithmic and analytic tools that reveal the interesting interplay in coded switches between packet placement, computation efficiency, and throughput performance.

%
%

\section{Full-Throughput Conditions and Placement Policies}
\label{sec:placement_policies}

In this section, we start with providing necessary and sufficient conditions for a \textit{full-throughput} solution, i.e., $L^*=L$ read packets. This is desired in practice to avoid delaying or reordering the read packets before fulfilling the read request. These conditions will be used later toward specifying packet placement policies, and analyzing the performance of read algorithms. Subsequently, we define the three policies this paper considers for placing packets in the switch memory: {\em uniform}, {\em cyclic}, and {\em design}.

\subsection{Full-throughput conditions}
\label{subsec:ft_conditions}

To find a \textit{necessary} condition for the existence of a full-throughput solution, note that each read packet requires at least $k$ MUs not used by any other packet. Thus, at least $kL$ MUs must be \textit{covered} by the requested packets, such that the following inequality
\begin{equation}
\label{coverage_condition}
\left| {\bigcup\limits_{i = 1}^L {{\mathcal{S}_i}} } \right| \ge kL
\end{equation}
must hold in any nkMTP instance with a full-throughput solution. We refer to \eqref{coverage_condition} as the \textit{coverage condition}. Note that when $k=n$ this condition (with equality) becomes sufficient as well, as the condition implies in this case that there is no contention between packets. We now move to find a \textit{sufficient} condition for the existence of a full-throughput solution. Let us extend the set notation to represent intersections of MU sets, that is, ${{\cal S}_{\cal I}} \triangleq \bigcap\limits_{j \in {\cal I}} {{{\cal S}_j}}$, for $\mathcal{I} \subseteq \left\{ {1,2,...,L} \right\}$.

\begin{theorem}
\label{th:max_overlap}
Let $\mathcal{S}_1,...,\mathcal{S}_L$ ($L \ge 2$) be the MU sets of an nkMTP instance. Then an $L^*=L$ solution exists if
\begin{equation}
\label{suff_cond}
\forall i,j:i \ne j,\hspace{3pt} \left| {{{\cal S}_i} \cap {{\cal S}_j}} \right| \le \frac{{2(n - k)}}{{L - 1}}\triangleq {t_{\max }}.
\end{equation}
\end{theorem}

\begin{proof}
Denote by ${\Phi_{s,\mathcal{L}}}$ the sum of cardinalities of intersections of $s$ distinct sets taken from the MU sets indexed by a certain set $\mathcal{L}$
\begin{equation}
\label{phi_sL}
{\Phi _{s,\mathcal{L}}} = \sum\limits_{\mathcal{I} \subseteq \mathcal{L},\left| \mathcal{I} \right| = s} \left| \mathcal{S}_\mathcal{I}\right| .
\end{equation}
As an example, if $\mathcal{L}$ is the set $\left\{ {1,2,3} \right\}$, then ${\Phi _{2,\mathcal{L}}}$ is $\left| {{\mathcal{S}_1} \cap {\mathcal{S}_2}} \right| + \left| {{\mathcal{S}_1} \cap {\mathcal{S}_3}} \right| + \left| {{\mathcal{S}_2} \cap {\mathcal{S}_3}} \right|$. As we saw in Section \ref{sec:read_problem}, an nkMTP instance can be represented as a bipratite graph with the packets and MUs being the disjoint vertex sets. In this representation, packet vertices need to be matched to disjoint sets of $k$ MU vertices. According to the extended Hall's theorem \cite{Viderman}, all the $L$ packet vertices can be matched (i.e., an $L^* = L$ solution exists) if and only if
\begin{equation}
\label{hall_cond}
\left| {\bigcup\limits_{j \in \mathcal{L}} {{\mathcal{S}_j}} } \right| \ge k\left| \mathcal{L} \right|
\end{equation}
for every subset $\mathcal{L}  \subseteq \left\{ {1,2,...,L} \right\}$. In words, at least $k|\mathcal{L}|$ distinct MUs should be present in each $\mathcal{L}$ sub-family of the $L$ MU sets. Using the inclusion-exclusion principle, \eqref{hall_cond} is equivalent to the requirement
\begin{equation}
\label{ix1}
n|\mathcal{L}| - {\Phi_{2,\mathcal{L}}} + \sum\limits_{s = 3}^{|\mathcal{L}|} {{{\left( { - 1} \right)}^{s - 1}}{\Phi_{s,\mathcal{L}}}} \ge k|\mathcal{L}|
\end{equation}
for every $\mathcal{L} \subseteq \left\{ {1,2,...,L} \right\}$, $\left| \mathcal{L} \right| \ge 2$ (for $\left| \mathcal{L} \right|=1$, \eqref{ix1} reduces to the requirement $n \ge k$ that always holds). The sum $\sum\limits_{s = 3}^{|\mathcal{L}|} {{{\left( { - 1} \right)}^{s - 1}}{\Phi_{s,\mathcal{L}}}}$ is non-negative, as it compensates for over-subtraction of pairwise intersection cardinalities in the inclusion-exclusion process. Therefore, \eqref{ix1} holds if the inequality
\begin{equation}
\label{phi1}
{\Phi_{2,\mathcal{L}}} \le |\mathcal{L}|\left( {n - k} \right)
\end{equation}
holds for every $\mathcal{L}$. We can bound ${\Phi_{2,\mathcal{L}}}$ by bounding the pairwise intersection cardinalities
\begin{equation}
\label{phi2}
{\Phi_{2,\mathcal{L}}} =  \sum\limits_{i \ne j \subseteq \mathcal{L}} {\left| {{{\cal S}_i} \cap {{\cal S}_j}} \right|} \le {{|\mathcal{L}|} \choose 2} {\max _{i \ne j \subseteq \mathcal{L}}}\left| {{\mathcal{S}_i} \cap {\mathcal{S}_j}} \right|.
\end{equation}
Finally, combining \eqref{phi2} and \eqref{phi1}, the inequality \eqref{phi1} holds when
\begin{equation}
\label{phi3}
{\max _{i \ne j \subseteq {\cal L}}}\left| {{{\cal S}_i} \cap {{\cal S}_j}} \right| \le \frac{2(n-k)}{|\mathcal{L}|-1}.
\end{equation}
 We now observe that the condition of the theorem \eqref{suff_cond} implies \eqref{phi3} because $|\mathcal{L}|\leq L$ for every $\mathcal{L}$. \qed
\end{proof}

We refer to condition \eqref{suff_cond} as the \textit{pairwise condition}. The full-throughput coverage and pairwise conditions above will serve us later for specifying placement policies and analyzing their throughput performance. The sufficient pairwise condition will give lower bounds on average throughput, and the necessary coverage condition will give upper bounds. We next turn to specify three placement policies for the switch write path: the {\em uniform}, {\em cyclic} and {\em design} placements. In subsequent sections these placement policies are given efficient read algorithms and performance analysis.

\subsection{Uniform placement}
\label{subsec:uniform_placement}
In the first placement policy we consider, the $n$ chunks of a packet may be placed in any set of $n$ MUs taken from the $N$ MUs in the system. That is, the set of a packet MU indices can be one of the ${N \choose n}$ $n$-subsets of $\mathcal{S}=\left\{ {0,1,...,N - 1} \right\}$. We term this policy as \textit{uniform} placement, but note that no probability distribution is assumed. This placement policy is the most general as no structure is imposed on the placement of packet chunks to memory.
\begin{example} \label{example:uniform_placement}
Assume that $N=5$ and $n=3$. There are ${N \choose n} = 10$ possible MU sets when the uniform placement policy is used: $\left\{ {0,1,2} \right\},\left\{ {0,1,3} \right\},\left\{ {0,1,4} \right\},\left\{ {0,2,3} \right\},\left\{ {0,2,4} \right\},\left\{ {0,3,4} \right\},\left\{ {1,2,3} \right\},\left\{ {1,2,4} \right\},\left\{ {1,3,4} \right\}$ and $\left\{ {2,3,4} \right\}$.
\end{example}
This placement policy is convenient to implement, because it has maximal flexibility to choose MUs to write based on load and available space. However, this comes with a price, as solving efficiently an arbitrary uniform placement instance amounts to solving nkMTP, shown to be NP-hard in Section \ref{sec:problem_formulation}. Therefore, in the rest of this section we propose two additional placement policies, which will be shown later to admit efficient optimal read algorithms.

\subsection{Cyclic placement}
\label{subsec:cyclic_placement}

In the second placement policy we propose, termed as \textit{cyclic} placement, we add a structure constraint on the MUs chosen to store packet chunks. The constraint is that the possible MU sets are composed of $n$ \textit{cyclic consecutive} MU indices. The number of possible MU sets is $N$ (assuming that $n<N$), which is smaller in all non-trivial cases than the ${N \choose n}$ sets in the uniform placement. An MU set in a cyclic instance can be conveniently thought of as an \textit{arc} covering $n$ cyclic consecutive points out of $N$ points on a circle, where the points are considered as MUs. An example for a circle-arc representation of a cyclic instance is shown in Fig. \ref{fig:cyclic_nkmtp_instance}.

\begin{figure}[t]
\centering
\includegraphics[scale=0.65]{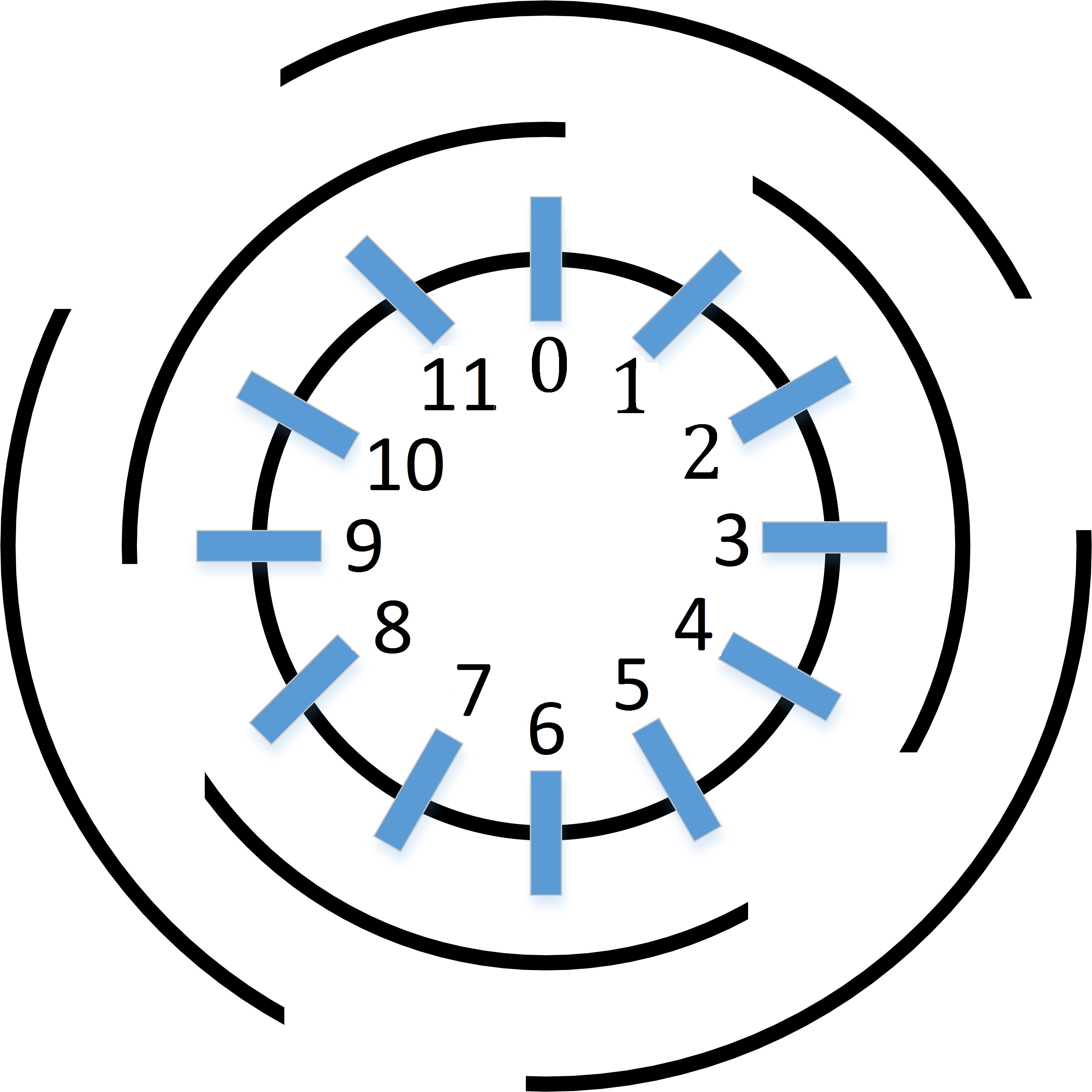}
\caption{A cyclic instance in a circle-arc representation. The marks on the inner circle represent $N=12$ MUs, where the $L=6$ outer arcs represent packets stored each in $n=4$ cyclic consecutive MUs.}
\label{fig:cyclic_nkmtp_instance}
\end{figure}

\begin{example}
Assume that $N=5$ and $n=3$. There are $N=5$ possible MU sets when the cyclic placement policy is used: $\left\{ {0,1,2} \right\},\left\{ {1,2,3} \right\},\left\{ {2,3,4} \right\},\left\{ {3,4,0} \right\}$ and $\left\{ {4,0,1} \right\}$. Note that these sets are contained in the sets of Example \ref{example:uniform_placement}.
\end{example}

A further restriction of the cyclic placement policy gives a simple placement policy where the $N$ MUs are statically partitioned to $N/n$ disjoint sets of $n$ consecutive MUs (assuming that $n$ divides $N$), and each packet is restricted to one of these sets. However, using the full cyclic (non partitioned) placement is beneficial for increased flexibility at the read path.


\subsection{Design placement}
\label{subsec:design_placement}

In the third policy we consider, our aim is to guarantee a full-throughput solution, i.e., $L^*=L$ read packets. Motivated by the sufficient condition of Theorem \ref{th:max_overlap}, we propose to construct a collection of MU sets with overlap at most $t_{\rm max}=2(n-k)/(L-1)$, using {\em combinatorial block designs}. To find such MU sets, we use the so called $t$-designs \cite{Lint} with carefully chosen parameters. A $t$-$(N, n, \lambda)$ design consists of $n$-element subsets (\textit{blocks}) taken from a set of $N$ elements, such that every $t$ elements taken from the set appear in exactly $\lambda$ subsets. $2$-designs are of particular interest in the literature, and they are known as \textit{balanced incomplete block design} (BIBD).

While it is not a trivial problem to construct combinatorial designs with arbitrary parameters, many design families are known within the vast literature on this topic \cite{Colbourn, Stinson}. When $\lambda= 1$, $t$-designs are known as \textit{Steiner systems}, and they contain (when exist) $b={N \choose t} / {n \choose t}$ blocks \cite{Lint}. Note the relation between $N,n,t$ and $b$, demonstrating that these values cannot be chosen arbitrarily. In general, a large value of $b$ is desired (i.e., large number of blocks) to have fewer occurrences where two requested packets use the same block as their MU set. We use the notation $t$-$(N,n)$ for Steiner systems (where $\lambda=1$ is implied). Our interest lies in block designs with $t = t_{\rm max}+1$ and $\lambda = 1$, such that the pairwise intersection cardinality is at most $t_{\rm max}$. We term a placement method where packets are constrained to such MU sets as \textit{design} placement. In such a placement, we are guaranteed the existence of an $L^*=L$ solution if the packets are placed in $L$ distinct MU sets.

\begin{example}
\label{BIBD_ex}
Consider the set $\mathcal{M} = \left\{ {0,1,2,3,4,5,6} \right\}$. Its subsets (blocks)
${\mathcal{M}_1} = \left\{ {0,1,2} \right\},{\mathcal{M}_2} = \left\{ {0,3,4} \right\},{\mathcal{M}_3} = \left\{ {0,5,6} \right\},{\mathcal{M}_4} = \left\{ {1,3,5} \right\},{\mathcal{M}_5} = \left\{ {1,4,6} \right\},{\mathcal{M}_6} = \left\{ {2,3,6} \right\}$ and ${\mathcal{M}_7} = \left\{ {2,4,5} \right\}$ form a $2$-$(7,3,1)$ BIBD (which is a Steiner system). There are  ${7 \choose 2} / {3 \choose 2}=7$ blocks in this design, known as the \textit{Fano plane} \cite{Stinson}. It can be seen that no two blocks intersect on more than one element (and each pair of elements is contained in exactly one block), such that this design can be used when $t_{\rm max}=1$ is desired. Since $n=3$ in this design, this value of $t_{\rm max}$ guarantees a solution for either $k=2$ and $L=3$ or $k=1$ and $L=5$.
\end{example}

An alternative for constructing MU sets with overlap at most $t_{\rm max}$ is the use of \textit{constant-weight codes}. A binary $\left( {N,d,n} \right)$ constant-weight code contains binary codewords of length $N$, each with $n$ non-zero coordinates, such that the Hamming distance between every two vectors (i.e., the number of coordinates in which they differ) is at least $d$. The supports (i.e., the non-zero coordinates) of the codewords form an $\left( {n - d/2 + 1} \right)$-$\left( {N,n,1} \right)$ \textit{packing} \cite{Ostergard}, in which each $\left(n - d/2 + 1\right)$-subset appears \textit{at most} once. A packing can be thought of as a relaxed version of a block design, which similarly satisfies pairwise intersection of at most $t_{\rm max}$ when setting $d = 2\left( {n - {t_{\max }}} \right)$. As a consequence, we can use $\left( {N,2\left( {n - {t_{{\text{max}}}}} \right),n} \right)$ constant-weight codes to construct MU sets with the desired $t_{\rm max}$ intersection property. As a large number of valid MU sets is desired, we are interested in constant-weight codes with the maximum possible number of codewords for the given parameters. Constructions of constant-weight codes and lower/upper bounds on the maximum number of codewords for certain parameters $N,d$ and $n$, denoted $A\left({N,d,n} \right)$, are provided e.g. in  \cite{Verdu, Agrell, Ostergard}.

\begin{example}
Consider the binary vectors of length $N=5$ with exactly $n=3$ non-zero coordinates. These vectors form an $\left( {N,1,n} \right)$ constant-weight code. The corresponding MU sets are the codeword supports, which appear in Example \ref{example:uniform_placement}.
\end{example}

\section{Read Algorithms}
\label{sec:read_algorithms}

As we saw in Section \ref{sec:problem_formulation}, nkMTP is intractable for general instances obtained when the uniform placement is used. In this section, we provide explicit and efficient optimal read algorithms for the cyclic and design policies.

\subsection{Cyclic placement}
\label{subsec:Cyclic_read}
In this subsection, we provide an efficient optimal algorithm for finding a maximum-throughput solution in the cyclic case. We start with the following important observation.

\begin{lemma}
\label{lemma_k_cons}
Assume an nkMTP instance with cyclic placement of the packet chunks. Then there exists an optimal solution where the $k$ MUs assigned to each read packet are cyclic consecutive.
\end{lemma}
\begin{proof}
We show that any optimal solution for the cyclic placement can be transformed into an optimal solution with a cyclic consecutive assignment of MUs to each packet. Assume an optimal solution with a gap in packet $j$'s assignment (i.e., the $k$ assigned MUs to packet $j$ are not cyclic consecutive). If the MUs in the gap are not assigned to any other packet, then clearly we can exchange MUs between the gap and the assigned MUs to obtain an assignment with no gap. Let us now consider a case where the MUs in the gap were assigned to other packets. Because of the cyclic placement and the fixed $n$, the packets assigned the gap MUs overlap with packet $j$ on either the MUs before the gap or the MUs after the gap. In either case we can exchange between MUs in the gap and MUs assigned to packet $j$ to obtain a cyclic consecutive assignment, as the MUs in the overlap can serve any of the overlapping packets. \qed
\end{proof}

Based on Lemma \ref{lemma_k_cons}, we propose an efficient algorithm for solving a cyclic instance. For convenience, we assume a circle-arc representation (see Section \ref{subsec:cyclic_placement}). Define an order of the packets with respect to packet $j_0$, such that the packets are sorted according to their arcs' starting points relatively to packet $j_0$'s starting point in clockwise order.
\begin{example}
Consider the cyclic instance in Fig. \ref{fig:cyclic_nkmtp_instance}, where the order is with respect to the topmost packet arc ($\{11,0,1,2\}$). The ordered packets are $\left\{ {11,0,1,2} \right\}$, $\left\{ {1,2,3,4} \right\}$, $\left\{ {3,4,5,6} \right\}$, $\left\{ {5,6,7,8} \right\}$, $\left\{ {7,8,9,10} \right\}$ and $\left\{ {9,10,11,0} \right\}$.
\end{example}
In the algorithm we begin with two empty sets $\Lambda$ and $\Omega$, which will eventually contain the read packets and their assigned MUs, respectively. We also initialize the sets ${\Lambda}_j$ and ${\Omega}_j$ (for $j=1,2,...,L$) as empty sets. The following algorithm solves optimally a cyclic nkMTP instance.

\begin{algorithm}
\label{Cyclic_algorithm}(Cyclic placement, optimal read algorithm)

For $j=1,2,...,L$, do:
\begin{enumerate}
\item Consider the set of packets $\left\{ {{\tilde{\mathcal S}_i}} \right\}_{i = 1}^L$  sorted with respect to packet $j$. Set $i:=1$.
\item \label{step_ij} If $|{\tilde{\mathcal{S}}}_{i}|\geq k$, add $i$ to $\Lambda_j$, and add the first $k$ MUs in ${\tilde{\mathcal{S}}}_{i}$ to $\Omega_j$. Remove the added MUs from all other packets.
\item Set $i:=i+1$. If $i \le L$, go to Step \ref{step_ij}. Otherwise, go to Step \ref{compare}.
\item If $\left| {{\Lambda _j}} \right| > \left| \Lambda  \right|$, set $\Lambda  := \Lambda_j$, $\Omega := \Omega_j$. \label{compare}
\end{enumerate}
\end{algorithm}

\begin{theorem}
The set of packets $\Lambda$ and their corresponding MUs in $\Omega$ found by Algorithm \ref{Cyclic_algorithm} are an optimal solution to a cyclic nkMTP instance.
\end{theorem}
\begin{proof}
According to Lemma \ref{lemma_k_cons}, there exists an optimal solution where the $k$ MUs assigned to a read packet are cyclic consecutive. We further show that without loss of generality, there is at least one packet $j_0$ in the solution that is assigned its first $k$ MUs. If there is no such packet, we can shift the solution counter-clockwise until this condition is met. Given that $j_0$ is a packet assigned its first $k$ MUs, we prove that the algorithm finds the optimal solution in iteration $j=j_0$. To prove this, we show that if packet $i$ was added to $\Lambda_{j_0}$ in Step \ref{step_ij}, then this packet appears in the optimal solution. This is proved by induction on $i$. Assume all packets $1,\ldots,i-1$ (in the order by $j_0$) can be chosen as in Step \ref{step_ij}. Then we show that the $i$-th packet can be chosen in the same way. We assume by contradiction that $|{\tilde{\mathcal{S}}}_{i}|\geq k$ and there is no optimal solution that contains packet $i$. Then we look at the smallest packet index $i'>i$ for which packet $i'$ appears in the optimal solution. From the fact that its $k$ assigned MUs are cyclic consecutive it is possible to shift the assignment to the first MU index in ${\tilde{\mathcal{S}}}_{i}$, and replace $i'$ by $i$ in the optimal solution without affecting the selection of packet indices larger than $i'$. This is a contradiction.

The proof is completed by observing that maximizing the size of the packet set $\Lambda_j$ over all indices $j$ is guaranteed to give the optimal solution, because at least one packet $j$ is qualified as a $j_0$ that is in the optimal solution with its first $k$ MUs.
\qed
\end{proof}

Algorithm \ref{Cyclic_algorithm} requires simple sorting and comparison operations, resulting in $\mathcal{O}(L^2)$ complexity. Since the solution method assures that MUs assigned to each read packet are cyclic consecutive, the non-assigned MUs can be regarded as $n-k$ erased symbols, or as a cyclic \textit{burst} of $n-k$ erasures. An example is shown in Fig. \ref{fig:cyclic_erasures}. This erasure structure suggests the use of $[n,k]$ \textit{binary cyclic} codes (not necessarily MDS), which are especially efficient at recovering burst erasures. Cyclic codes are linear codes, with the property that a cyclic shift of a codeword produces a codeword as well. These codes are capable of recovering from any cyclic burst erasure of length up to $n-k$ \cite{Ryan}. The use of binary cyclic codes simplifies the coding process considerably. The reason is that non-trivial MDS codes require non-binary field arithmetic and impose certain restrictions on the code parameters, which can mostly be lifted once cyclic binary codes are used.

\begin{figure}[t!]
\centering
\includegraphics[scale=0.9]{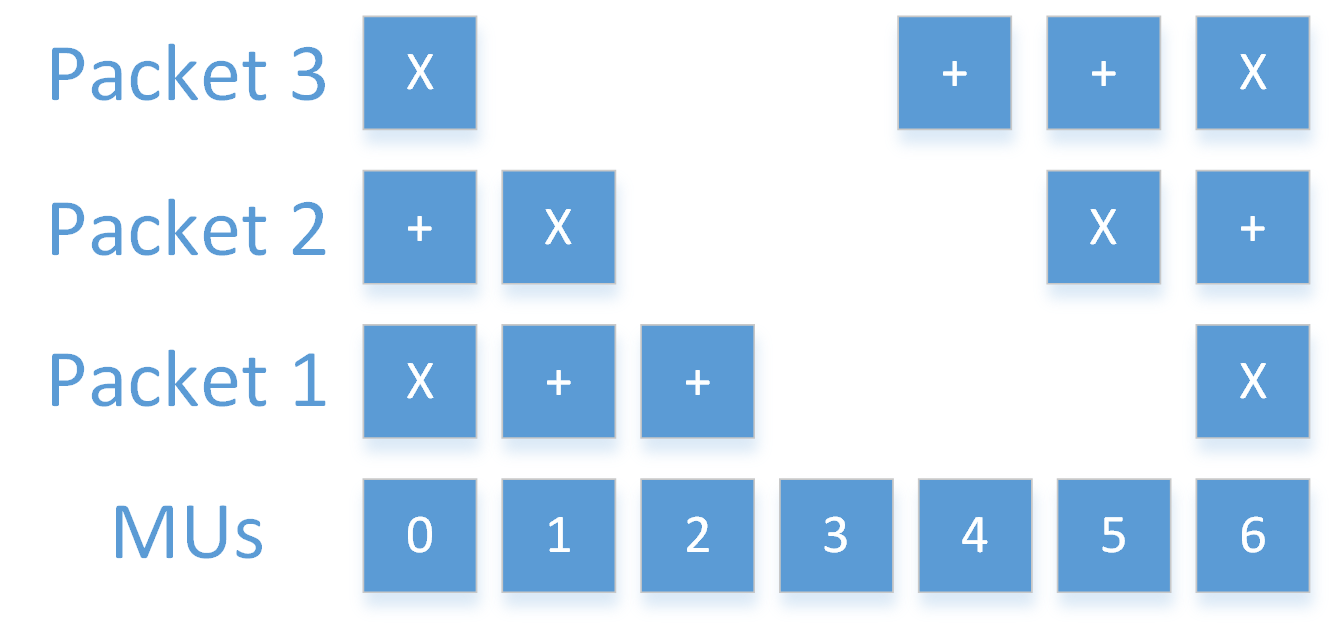}
\caption{A solution of a cyclic instance ($k=2, n=4$). '+' denotes an MU assigned to the packet where 'X' denotes an MU not assigned to the packet (erasure).}
\label{fig:cyclic_erasures} \vspace{-5pt}
\end{figure}

\begin{example}
Consider the (systematic) cyclic code ${\cal C} = \left\{ {0000,0101,1010,1111} \right\}$, where $k=2$ bits are encoded to $n=4$ bits. There are $n$ possible cyclic bursts of length $n-k=2$. Assume a burst in the first $2$ codeword positions. Then the remaining bits at the last $2$ positions are all distinct: $00, 01, 10$ and $11$, determining uniquely the codeword. The same holds for any cyclic burst erasure of length $n-k=2$.
\end{example}

\subsection{Design placement}

We now turn to provide an efficient optimal read algorithm for an nkMTP instance of the design placement. The algorithm we propose owes its efficiency to the sufficient pairwise condition satisfied by design-placement instances. As we show below, if the pairwise condition is satisfied, an optimal solution does not need to assign MUs contained in sets of {\em more than two} packets. This fact turns out to imply an extremely simple assignment algorithm. In the typical case $k>n/2$, a certain set (design block) of $n$ MUs can not serve more than one packet. Thus, we consider design instances with $L$ packets stored in $L$ \textit{distinct} MU sets (otherwise a subset of the packets stored in distinct blocks is considered). Denote by ${\mathop \mathcal{S}\limits^.}_{i}$ the MUs indexed in $\mathcal{S}_i$ that are not shared by any other MU set, and by ${\mathop \mathcal{S}\limits^.}_{ij}$ the MUs indexed in both $\mathcal{S}_{i}$ and $\mathcal{S}_j$ but not in any other MU set. In the rest of this sub-section, we show that for each $i$, at least $k$ MUs from the sets ${\mathop \mathcal{S}\limits^.}_{i}$ and ${\mathop \mathcal{S}\limits^.}_{ij}$ ($j \ne i$) can be assigned to packet $i$ (such that these MUs do not serve any other packet). At a high level, the assignment that guarantees $k$ MUs to packet $i$ is all of ${\mathop \mathcal{S}\limits^.}_{i}$, and half of each ${\mathop \mathcal{S}\limits^.}_{ij}$. We present this more formally in the following Algorithm~\ref{design_nkMTP_algorithm}. The algorithm is initialized with empty sets $\mathcal{S}'_i$ ($i=1,2,...,L$) that will eventually contain an optimal assignment of MUs to the packets. We use the notation $\left\lfloor x \right\rfloor$ (resp. $\left\lceil x \right\rceil$) for the floor (resp. ceiling) value of $x$, i.e., the largest integer not greater than $x$ (resp. the smallest integer not smaller than $x$).

\begin{algorithm}
\label{design_nkMTP_algorithm}(Design placement, optimal read algorithm)

For each packet $i=1,...,L$, do:
\begin{enumerate}
\item Add the MUs in ${\mathop \mathcal{S}\limits^.}_{i}$ to $\mathcal{S}'_i$.
\item For each $j$ such that $\left| {{{\mathop {\cal S}\limits^. }_{ij}}} \right|$ is even, add $\left| {{{\mathop {\cal S}\limits^. }_{ij}}} \right|/2$ MUs from ${{{\mathop {\cal S}\limits^. }_{ij}}}$ to $\mathcal{S}'_i$ (disjoint from the MUs added to $\mathcal{S}'_j$ in a different iteration).
\item For each $j$ such that $\left| {{{\mathop {\cal S}\limits^. }_{ij}}} \right|$ is odd, add either $\left\lceil {\left| {{{\mathop {\cal S}\limits^. }_{ij}}} \right|/2} \right\rceil$ or $\left\lfloor {\left| {{{\mathop {\cal S}\limits^. }_{ij}}} \right|/2} \right\rfloor$ MUs from ${{{\mathop {\cal S}\limits^. }_{ij}}}$ to $\mathcal{S}'_i$, according to the policy specified below under {\em floor/ceil balancing}. \label{floor_ceil}

\end{enumerate}
\end{algorithm}

\textbf{Floor/ceil balancing.} For Algorithm~\ref{design_nkMTP_algorithm} we need to specify whether to assign $\left\lceil {\left| {{{\mathop {\cal S}\limits^. }_{ij}}} \right|/2} \right\rceil$ or $\left\lfloor {\left| {{{\mathop {\cal S}\limits^. }_{ij}}} \right|/2} \right\rfloor$ in the odd case (step \ref{floor_ceil}). We show such an assignment that for each $i$ balances the number of floors and ceils sufficiently to guarantee at least $k$ assigned MUs. Construct an undirected graph $U$ whose vertices are the packet indices, and connect two vertices $i$ and $j$ by an edge if $\left| {{{\mathop \mathcal{S}\limits^. }_{ij}}} \right|$ is odd. Remove from the graph vertices not connected by an edge to any other vertex. An \textit{orientation} of $U$ is an assignment of a direction to each edge in $U$ (leading to a directed graph). There always exists an orientation of an undirected graph such that the number of edges entering and exiting every vertex differ by at most one \cite{Nash}. This orientation can be found in time linear in the number of edges \cite{Bondy}. We denote such an orientation by $\mathop U\limits^ \to$, and an example is shown in Fig. \ref{fig:balanced_orientation}. Given $\mathop U\limits^ \to$, we can rewrite step \ref{floor_ceil} in Algorithm~\ref{design_nkMTP_algorithm} in a precise way
\begin{enumerate}
\setcounter{enumi}{2}
\item For each $j$ such that $\left| {{{\mathop {\cal S}\limits^. }_{ij}}} \right|$ is odd
\begin{enumerate}
\item If the edge between $i$ and $j$ is oriented towards $i$ in $\mathop U\limits^ \to$, add  $\left\lceil {\left| {{{\mathop {\cal S}\limits^. }_{ij}}} \right|/2} \right\rceil$ MUs (not added earlier) from ${{{\mathop {\cal S}\limits^. }_{ij}}}$ to packet $i$.
\item Otherwise, add  $\left\lfloor {\left| {{{\mathop {\cal S}\limits^. }_{ij}}} \right|/2} \right\rfloor$ MUs (not added earlier) from ${{{\mathop {\cal S}\limits^. }_{ij}}}$ to packet $i$.
\end{enumerate}
\end{enumerate}

 \begin{figure}[t]
\centering
\includegraphics[scale=0.65]{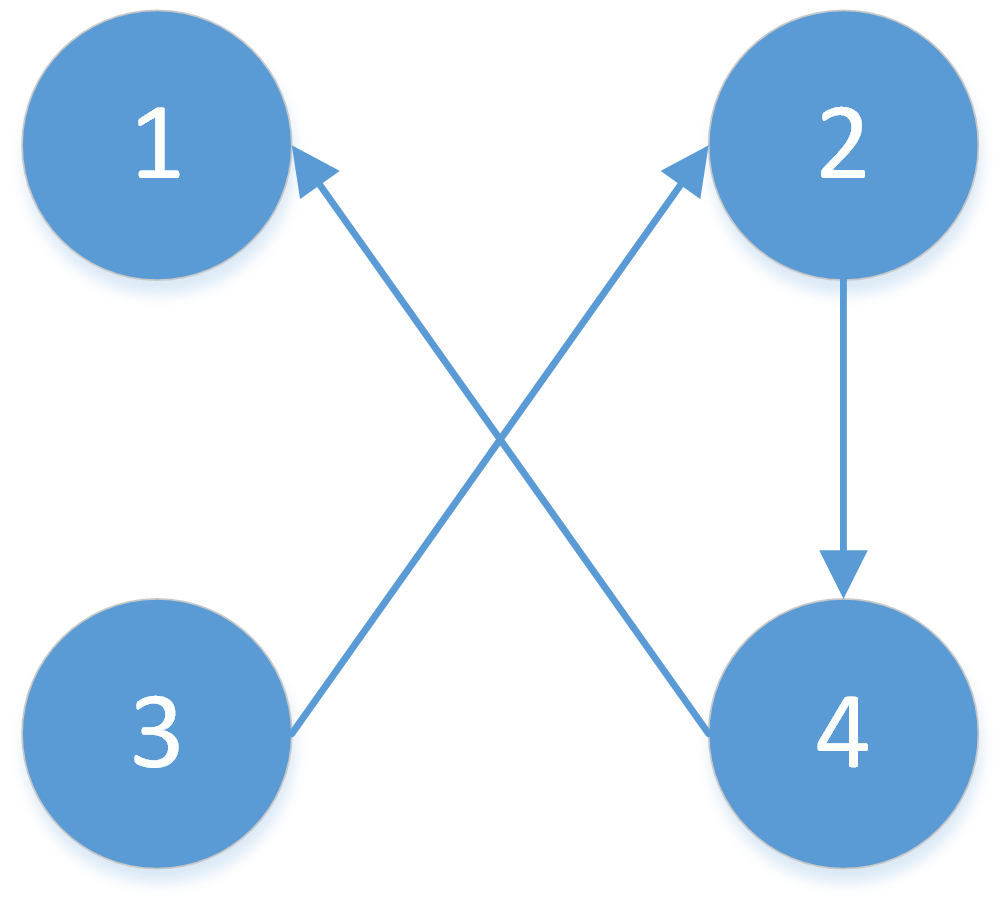}
\caption{An orientation where the number of edges entering and existing a vertex differ by at most once.}
\label{fig:balanced_orientation}
\end{figure}

\begin{theorem}
\label{thm:design_algorithm}
The sets $\mathcal{S}'_i$ in Algorithm \ref{design_nkMTP_algorithm} form an optimal solution to a design instance.
\end{theorem}
\begin{proof}
First, an MU added to $\mathcal{S}'_i$ is not added to any $\mathcal{S}'_{j \ne i}$. The reason is that MUs in ${\mathop \mathcal{S}\limits^.}_{i}$ are added to  $\mathcal{S}'_i$ only, and two \textit{disjoint} subsets of MUs (using complementary ceiling/floor operations) are taken from ${{{\mathop {\cal S}\limits^. }_{ij}}}$ to $\mathcal{S}'_i$ and $\mathcal{S}'_j$ only. Define the function ${f_{ij}}\left( x \right)$ as ${\left\lceil x \right\rceil }$ if the edge between $i$ and $j$ is oriented towards $i$ in $\mathop U\limits^ \to$, and ${\left\lfloor x \right\rfloor }$ otherwise. If $i$ and $j$ are not connected in $\mathop U\limits^ \to$ (i.e., $\left| {{{\mathop \mathcal{S}\limits^. }_{ij}}} \right|$ is even), ${f_{ij}}\left( x \right)$ is simply $x$. The cardinality of $\mathcal{S}'_i$ is then
\begin{equation}
\label{num_algo}
\left| {{{\mathop {\cal S}\limits^. }_i}} \right| + \sum\limits_{j \ne i} {{f_{ij}}\left( {\left| {{{\mathop {\cal S}\limits^. }_{ij}}} \right|}/2 \right)}.
\end{equation}
In the rest of this proof, we show that \eqref{num_algo} is lower-bounded by $k$. For odd-cardinality sets ${{{\mathop \mathcal{S}\limits^. }_{ij}}}$, $\left\lfloor {\left| {{{\mathop \mathcal{S}\limits^. }_{ij}}} \right|/2} \right\rfloor$ equals $\left| {{{\mathop \mathcal{S}\limits^. }_{ij}}} \right|/2 - 1/2$, and $\left\lceil {\left| {{{\mathop \mathcal{S}\limits^. }_{ij}}} \right|/2} \right\rceil$ equals ${\left| {{{\mathop {{\rm{ }}\mathcal{S}}\limits^. }_{ij}}} \right|/2} + 1/2$. Since the number of edges entering and exiting a vertex differ by at most one, the number of floor operations in \eqref{num_algo} might exceed the number of ceiling operations by at most one. Therefore, \eqref{num_algo} is lower-bounded by
\begin{equation}
\label{lb1}
\left| {{{\mathop {\cal S}\limits^. }_i}} \right| + {\frac{1}{2}\sum\limits_{j \ne i} {\left| {{{\mathop {\cal S}\limits^. }_{ij}}} \right|} } - \frac{1}{2}.
\end{equation}
According to the inclusion-exclusion principle,
\begin{equation}
\label{ix_Si}
\left| {{{\dot {\mathcal{S}}}_i}} \right| = \sum\limits_{\mathcal{J} \supseteq \left\{ i \right\}} {{{\left( { - 1} \right)}^{\left| \mathcal{J} \right| - 1}}\left| {{\mathcal{S}_\mathcal{J}}} \right|},
\end{equation}
\begin{equation}
\label{ix_Sij}
\left| {{{\dot {\mathcal{S}}}_{ij}}} \right| = \sum\limits_{\mathcal{J} \supseteq \left\{ {i,j} \right\}} {{{\left( { - 1} \right)}^{\left| \mathcal{J} \right|}}\left| {{\mathcal{S}_\mathcal{J}}} \right|}.
\end{equation}
Substitute $\left| {{{\dot {\mathcal{S}}}_i}} \right|$ and $\left| {{{\dot {\mathcal{S}}}_{ij}}} \right|$ in \eqref{lb1} by the sums expanding them in \eqref{ix_Si}-\eqref{ix_Sij}. Each set $\mathcal{J} \supseteq \left\{ i \right\}$ appears in the combined sums once due to $\left| {{{\dot {\mathcal{S}}}_i}} \right|$, and additional  $|\mathcal{J}|-1$ times (weighted by $1/2$ and with an opposite sign) due to the summation of $\left| {{{\dot {\mathcal{S}}}_{ij}}} \right|$ over $j \ne i$. Therefore, \eqref{lb1} equals
\begin{align}
&\frac{1}{2}\sum\limits_{\mathcal{J} \supseteq \left\{ i \right\}} {{{\left( { - 1} \right)}^{\left| \mathcal{J} \right| - 1}}\left( {3 - \left| \mathcal{J} \right|} \right)\left| {{\mathcal{S}_\mathcal{J}}} \right|} - \frac{1}{2}
\nonumber\\ \label{eq1}
&=n - \frac{1}{2}\sum\limits_{j \ne i} {\left| {{\mathcal{S}_{ij}}} \right|}  + \frac{1}{2}\sum\limits_{\scriptstyle{\cal J} \supseteq \left\{ i \right\},\hfill\atop
\scriptstyle|{\cal J}| \ge 4\hfill} {{{\left( { - 1} \right)}^{\left| {\cal J} \right|}}\left( {\left| {\cal J} \right| - 3} \right)\left| {{{\cal S}_{\cal J}}} \right|}  - \frac{1}{2}.
\end{align}

We claim that the last sum in \eqref{eq1} is non-negative. This sum counts the number of occurrences of MUs in intersection sets of $4$ packets or more that include packet $i$, multiplied by the factor $(|\mathcal{J}|-3)/2$, and with alternating signs. Consider a certain MU shared by exactly $T \ge 4$ packets including packet $i$. This MU appears in ${T-1 \choose |\mathcal{J}|-1}$ intersection sets of cardinality $4 \le |\mathcal{J}| \le T$ (we subtract $1$ as the packet index $i$ is always contained in $\mathcal{J}$). Therefore, the contribution of this MU to the count is
\begin{align}
&\frac{1}{2}\sum\limits_{|\mathcal{J}| = 4}^T {{{\left( { - 1} \right)}^{|\mathcal{J}| }}}(|\mathcal{J}|-3){T-1 \choose |\mathcal{J}|-1}
\\\nonumber
&=\frac{1}{2}\sum\limits_{|\mathcal{J}| = 3}^{T-1} {{{\left( { - 1} \right)}^{|\mathcal{J}|+1 }}}(|\mathcal{J}|-2){T-1 \choose |\mathcal{J}|}
\\\nonumber
&= \frac{1}{2}\sum\limits_{|\mathcal{J}| = 0}^{2} {{{\left( { - 1} \right)}^{|\mathcal{J}| }}}(|\mathcal{J}|-2){T-1 \choose |\mathcal{J}|}= (T-3)/2 \ge 0,
\end{align}
where we used the binomial identities
\begin{equation}
\sum\limits_{j = 0}^T {{{\left( { - 1} \right)}^{j }}}{T \choose j}=\sum\limits_{j = 0}^T j{{{\left( { - 1} \right)}^{j }}}{T \choose j}=0.
\end{equation}
This establishes the non-negativity of the last sum in \eqref{eq1}. ${\left| {{\mathcal{S}_{ij}}} \right|} \le t_{\rm max}$, so we conclude that \eqref{eq1} and thus \eqref{num_algo} are lower-bounded by
\begin{align}
\label{eq2}
n - \frac{1}{2}\sum\limits_{j \ne i} {{t_{\max }} -\frac{1}{2} = n - \frac{1}{2}\left( {L - 1} \right){t_{\max }} -\frac{1}{2} = k-\frac{1}{2}}.
\end{align}
The number of MUs added to packet $i$ in \eqref{num_algo} is necessarily integer. Thus, we have the integer value \eqref{num_algo} lower-bounded by the non-integer value \eqref{eq2}. This means that \eqref{num_algo} is in fact lower-bounded by the ceiling value of \eqref{eq2}, i.e., by $k$. \qed
\end{proof}

Algorithm \ref{design_nkMTP_algorithm} requires the construction of the sets ${\mathop \mathcal{S}\limits^.}_{i}$ and ${\mathop \mathcal{S}\limits^.}_{ij}$, which can be performed in $\mathcal{O}(Ln)$ operations by running over the elements in the sets $\mathcal{S}_i$. We then have to find a balanced path in a graph that is complete in the worst case (i.e., when $|{\mathop \mathcal{S}\limits^.}_{ij}|$ are all odd), with $\mathcal{O}(L^2)$ edges. The path is found in linear-time in the number of edges, such that the total complexity of Algorithm \ref{design_nkMTP_algorithm} is $\mathcal{O}(L(n+L))$.

\begin{example}
Consider the block design of Example \ref{BIBD_ex} (where $n=3$). This design can be used to read $L=3$ packets if $t_{\rm max} = n-k =1$, i.e., when $k=2$. Assume that the three packets are stored in the MU sets ${{\cal S}_1} = \left\{ {1,2,3} \right\}$, ${{\cal S}_2} = \left\{ {1,4,5} \right\}$ and ${{\cal S}_3} = \left\{ {3,5,6} \right\}$. Then ${\mathop \mathcal{S}\limits^.}_{1} = \left\{ 2 \right\},{\mathop \mathcal{S}\limits^.}_{2} = \left\{ 4 \right\},{\mathop \mathcal{S}\limits^.}_{3} = \left\{ 6 \right\},{\mathop \mathcal{S}\limits^.}_{12} = \left\{ 1 \right\},{\mathop \mathcal{S}\limits^.}_{13} = \left\{ 3 \right\}$ and ${\mathop \mathcal{S}\limits^.}_{23} = \left\{ 5 \right\}$. Since all the pairwise sets are of odd cardinality, $U$ in this case is a complete graph with three vertices. Labeling these vertices $1,2$ and $3$, a valid orientation $\mathop U\limits^ \to$ is $1 \to 2 \to 3 \to 1$. Using Algorithm \ref{design_nkMTP_algorithm}, we obtain ${{\cal S}'_1} = \left\{ {2,3} \right\},{{\cal S}'_2} = \left\{ {1,4} \right\}$ and ${{\cal S}'_3} = \left\{ {5,6} \right\}$.
\end{example}

\section{Probabilistic analysis}
\label{sec:prob_analysis}

In this section, we consider \textit{ensembles} of random instances characterized by $k,n,N,L$ and the placement policy in use, where an instance is obtained by a random placement of the $L$ packets. Our primary objective is to calculate or bound the full-throughput probability $\Pr \left( {{L^*} = L} \right)$ for the three placement policies discussed above. For the uniform and cyclic placements we use the coverage and pairwise conditions (see Section \ref{subsec:ft_conditions}) to obtain upper and lower bounds, respectively, on the full-throughput probability. We later present a convenient tighter probabilistic framework for analyzing the throughput performance for the design placement.

\subsection{Uniform placement}
\label{subsec:uniform_prob}

Denote the probability of the coverage condition \eqref{coverage_condition} in the uniform placement by $p_{\rm cover}^{\rm uni}$. The full-throughput probability $\Pr \left( {{L^*} = L} \right)$ is clearly upper bounded by $p_{\rm cover}^{\rm uni}$. The coverage condition in the uniform case is equivalent to the requirement that the union of $L$ random $n$-subsets of an $N$-element set results in a set of cardinality at least $kL$. A closed-form expression for this probability is provided by the union model \cite{QPEC}, which is an extension of the balls-and-bins model \cite{Mitz}. The details are provided in Appendix \ref{appendix:p_cover}. Through this calculation we obtain an upper bound on the full-throughput probability for any combination of $k,n,N$ and $L$. Exact calculation of the full-throughput probability for the uniform placement seems hard, and even a lower bound through the pairwise condition~\eqref{suff_cond} is not available. This lack of positive results for the uniform placement is not very surprising given the computational hardness of solving it optimally.

\subsection{Cyclic placement}

Considering the circle-arc representation of cyclic instances (see Section \ref{subsec:cyclic_placement}), the probability of the coverage condition is the probability that at least $kL$ points of the circle are covered by $L$ random arcs of $n$ cyclic consecutive points each. In \cite{Holst, Barlevy}, the probability distribution of the number of vacant points on a circle once $L$ random arcs are placed without replacement was derived. In our case, replacement is allowed (i.e., the same MU set may serve two packets or more), and for the upper bound on full-throughput probability we are actually interested in the complement distribution of the occupied points. The details are provided in Appendix \ref{appendix:p_cover}. We denote the probability of the coverage condition in the cyclic case by $p_{\rm cover}^{\rm cyc}$. For the lower bound, unlike the uniform policy, the structure in the cyclic case allows to find the probability of the pairwise condition, which we denote $p_{\rm pair}^{\rm cyc}$.
\begin{theorem}
\label{th:cyclic_pair_prob}
Consider an instance drawn at random from a cyclic ensemble with parameters $N,n,L$. The probability that the maximum pairwise intersection cardinality is at most $t_{\rm max}$ is
\begin{equation}
\label{p_cyc_pair}
p_{{\rm{pair}}}^{{\rm{cyc}}} = {N^{1 - L}}\prod\limits_{i = 1}^{L - 1} {\left( {N - L\left( {n - {t_{{\rm{max}}}}} \right) + i} \right)}.
\end{equation}
\end{theorem}
\begin{proof}
Consider a circle-arc representation of the cyclic nkMTP instances. Assume clockwise order, and that each packet arc does not precede the first packet arc. Each placed packet prevents the placement of the start of any other packet in its first $n-t_{\rm max}$ MUs. In a legal placement (i.e., when the pairwise intersection cardinality is at most $t_{\rm max}$), there are $N-L(n-t_{\rm max})$ MUs that do not belong to the first $n-t_{\rm max}$ MUs of any packet. Thus, the number of legal placements (given the order constraint above) is equivalently the number of ways to partition $N-L(n-t_{\rm max})$ MUs to $L$ sets of cyclic consecutive MUs. Thinking of the latter MU sets as \textit{gaps}, they can be distributed in ${N-L(n-t_{\rm max}) + L -1 \choose L-1}$ ways, which is the number of $L$ non-negative integers (gap lengths) whose sum is $N-L(n-t_{\rm max})$ \cite{Lint}. Each legal placement is obtained (uniquely) as a combination of 1) the starting MU for the first drawn packet, 2) a gap configuration, and 3) a permutation of the other $L-1$ packets. Hence to get the total number of legal placements we multiply the number of gap configurations by $N$ (the number of possible starting points for the first packet) and by $(L-1)!$ (the number of permutations of $L-1$ packets). After normalizing by the total number of (legal and illegal) placements $N^L$, we obtain \eqref{p_cyc_pair}. \qed
\end{proof}

\subsection{Design placement}
\label{subsec:design_prob}
The design placement enjoys a sharper characterization of full-throughput instances, which simplifies the probabilistic analysis. It is sufficient that the $L$ MU sets are different design blocks, and the request is full-throughput by the design properties and the sufficient pairwise condition \eqref{suff_cond}. Thus, the probability that a random design instance contains a full-throughput solution is lower-bounded by the probability that the $L$ MU sets are distinct. We denote this probability by $p^{\rm des}_{\rm pair}$. To find this probability, we use the balls-and-bins model \cite{Mitz}. In this model, there are $L$ balls and $b$ bins (recall that $b$ is the number of blocks in the design), where the balls are placed independently and uniformly at random in the bins. The probability of $L$ distinct blocks is the probability of $L$ non-empty bins \cite{Mitz}, which equals
\begin{align}
\label{balls_and_bins}
p^{\rm des}_{\rm pair} = {b \choose L} \frac{1}{{{b^L}}}\sum\limits_{j = 0}^{L} {{{\left( { - 1} \right)}^j}{L \choose j}{{\left( {L - j} \right)}^L}}.
\end{align}
In the typical case $k>n/2$, each block can serve only one packet, and thus $p^{\rm des}_{\rm pair}$ is the {\em exact} probability of a full-throughput solution in the design case.

To demonstrate the possible improved performance when the design placement is used, assume that $n=k+1$ for a fixed $k$ value. If the desired number of read packets is $L^*=L=3$, we can take the $2$-$(k^2+k+1,k+1)$ Steiner system \cite{Lint}, where the sufficient pairwise condition $t_{\rm max}=n-k=1$ is guaranteed by the $t=2$ parameter of the design. To have a full-throughput solution we need that the $L=3$ blocks drawn from the $b=k^2+k+1$ blocks of the design will be all distinct. In Fig. \ref{fig:cyclic_design}, we plot $p^{\rm des}_{\rm pair}$, which is $\Pr(L^*=L)$ in the design case, in comparison to the uniform upper bound and the cyclic lower and upper bounds on $\Pr(L^*=L)$ ($p_{\rm cover}^{\rm uni}$, $p_{\rm pair}^{\rm cyc}$ and $p_{\rm cover}^{\rm cyc}$, respectively). The exact probability for an uncoded cyclic placement is found using the coverage condition \eqref{coverage_condition}, and in this case it coincides with $p_{\rm cover}^{\rm cyc}$. We also plot $p_{\rm sim}^{\rm cyc}$, the only graph in Fig. \ref{fig:cyclic_design} obtained using simulations, which is the empirical $\Pr(L^*=L)$ in the cyclic case. The results clearly demonstrate that the design policy exhibits significantly superior performance. It is shown that with a fixed redundancy of $1$ chunk per packet, the full-throughput probability of the design placement grows monotonically when $k$ grows and $N=k^2+k+1$ MUs are deployed in the switch.

\begin{figure}[t]
\centering
\includegraphics[scale=0.6]{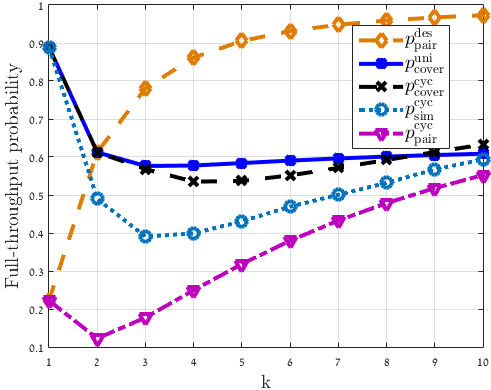}
\caption{A comparison of full-throughput performance bounds ($n=k+1, N=k^2+k+1$).}
\label{fig:cyclic_design}
\end{figure}

\section{Simulation Results}
\label{sec:simulations}

In this section, we provide simulation results of the average throughput performance of the placement policies proposed in Section \ref{sec:placement_policies}. Recall from Section~\ref{sec:read_problem} that the average throughput equals a constant times the average $L^*$ of ensemble instances. Hence evaluating the average throughput can be done by solving random instances optimally, and averaging the resulting $L^*$ values empirically. For the uniform placement we solved uniform nkMTP instances by an exhaustive-search algorithm (recall that no efficient algorithm is likely to exist in this case, see Section \ref{sec:problem_formulation}). We compare it to a greedy (suboptimal) solution, where random packets are assigned $k$ MUs, until no $k$ MUs that can serve a packet remain. To solve cyclic instances optimally, we used Algorithm \ref{Cyclic_algorithm}. A comparison of the average throughput performance $\bar \rho$ (see Section \ref{sec:read_problem}) of the uniform and cyclic placement policies is provided in Fig. \ref{tp_results} for $k=3$ and $n=3$ (uncoded case) up to $n=6$. Several observations follow from these results. First, coding improves throughput performance considerably. Taking the cyclic case as an example when $L=4$, the throughput performance is improved by $18\%$ ($n=4$) to $52\%$ ($n=6$) compared to the uncoded case ($n=3$). Another important observation is that the uniform policy does \textit{not} necessarily lead to better performance compared to the cyclic policy. Actually, the relation between the performance of these schemes depends on the system parameters. We can see that the cyclic case provides higher throughput performance when $k$ is close to $n$, showing that the structure becomes helpful when the read-flexibility in choosing $k$ MUs decreases. On the other hand, when the redundancy becomes larger (i.e., when $n$ becomes large compared to $k$), the uniform and the cyclic placement policies exhibit similar performance, each with a slight advantage at different $L$ values. When computational complexity is taken into account, the cyclic placement becomes superior over uniform, because it outperforms the low-complexity greedy read algorithm.

\begin{figure*}
        \centering
        \begin{subfigure}[b]{0.46\textwidth}
                \includegraphics[scale=0.57]{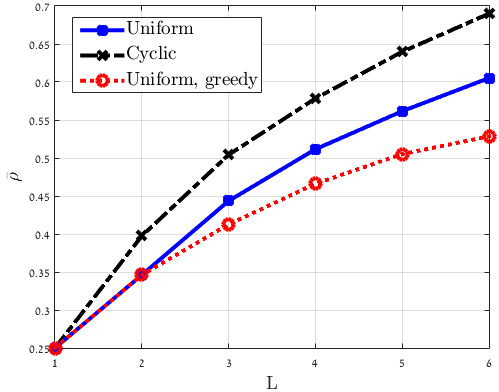}
                \caption{$k=n=3$ (no coding).}
                \label{k3n3_tp}
        \end{subfigure}%
        ~ 
        \begin{subfigure}[b]{0.46\textwidth}
                \includegraphics[scale=0.57]{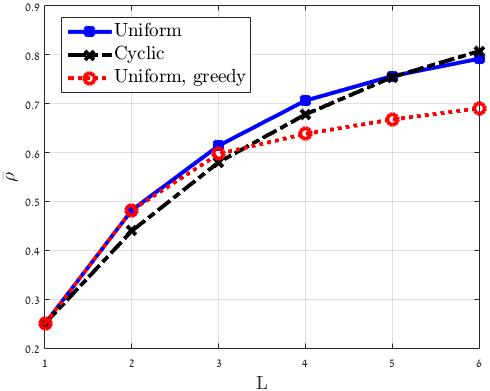}
                \caption{$k=3, n=4.$}
                \label{k3n4_tp}
        \end{subfigure}
        ~ 

\vspace{10pt}
        \begin{subfigure}[b]{0.46\textwidth}
                \includegraphics[scale=0.57]{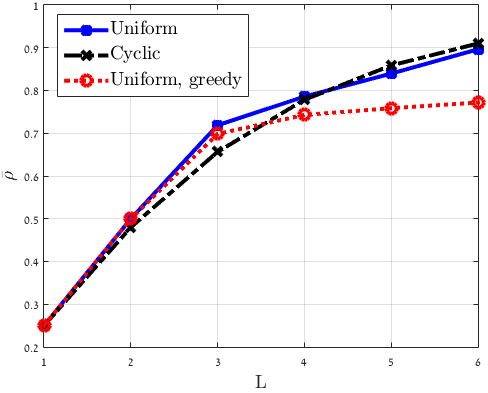}
                \caption{$k=3, n=5.$}
                \label{k3n5_tp}
        \end{subfigure}
             \begin{subfigure}[b]{0.46\textwidth}
                \includegraphics[scale=0.57]{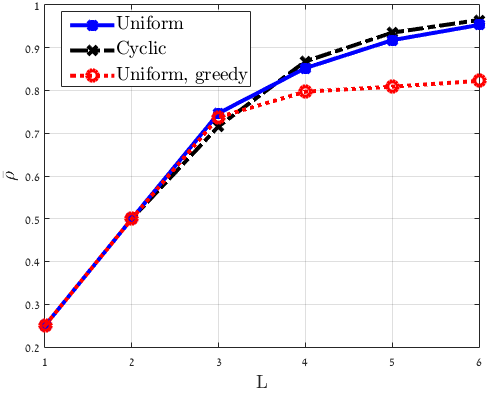}
                \caption{$k=3, n=6.$}
                \label{k3n6_tp}
        \end{subfigure}

        \caption{Average throughput $\bar \rho$ performance comparison ($N=12$). \label{tp_results}
}
\end{figure*}

Another performance measure we investigated is the number of packets that are read (i.e., the value of $L^*$) with high probability (w.h.p.) in a random instance. In Fig. \ref{whp_results}, we show $L^*$ values that were observed with probability at least $0.95$. Similarly to the results in Fig. \ref{k3n3_tp}, the cyclic scheme provides better $L^*$ performance in the uncoded case. For moderate $n$ values, the uniform policy is better than cyclic, where for larger $n$ values the performance becomes close. Regardless of the placement policy in use, coding improves the throughput performance. For example, when $k=n=3$ (Fig. \ref{k3n3_whp}) and the load is $L=3$, we expect to get w.h.p. only one packet at the output. On the other hand, when coding is introduced such that $k=3$ and $n=4$, this increases to $L^*=2$ packets and keeps improving up to $L=L^*=3$ for $k=3$ and $n=6$. That is, when the switch is required to fulfill all $L$ requests w.h.p., coding is an important tool.

\begin{figure*}
        \centering
        \begin{subfigure}[b]{0.46\textwidth}
                \includegraphics[scale=0.57]{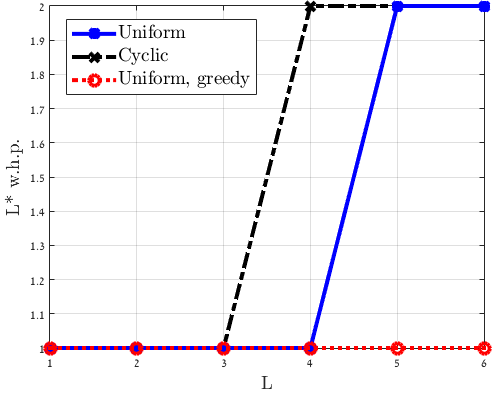}
                \caption{$k=n=3$ (no coding).}
                \label{k3n3_whp}
        \end{subfigure}%
        ~ 
        \begin{subfigure}[b]{0.46\textwidth}
                \includegraphics[scale=0.57]{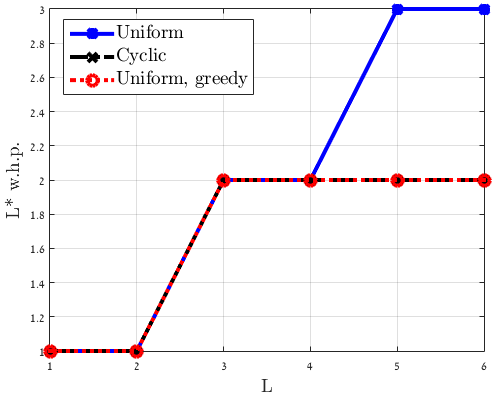}
                \caption{$k=3, n=4.$}
                \label{k3n4_whp}
        \end{subfigure}
        ~ 

\vspace{10pt}
        \begin{subfigure}[b]{0.46\textwidth}
                \includegraphics[scale=0.57]{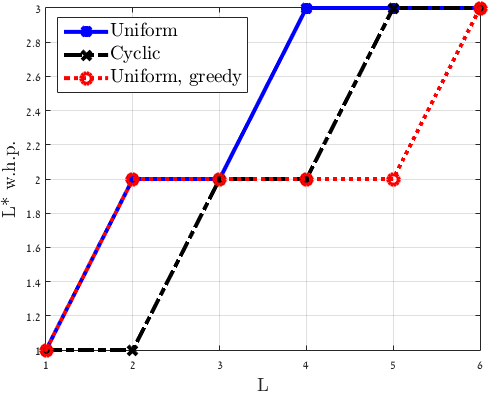}
                \caption{$k=3, n=5.$}
                \label{k3n5_whp}
        \end{subfigure}
             \begin{subfigure}[b]{0.46\textwidth}
                \includegraphics[scale=0.57]{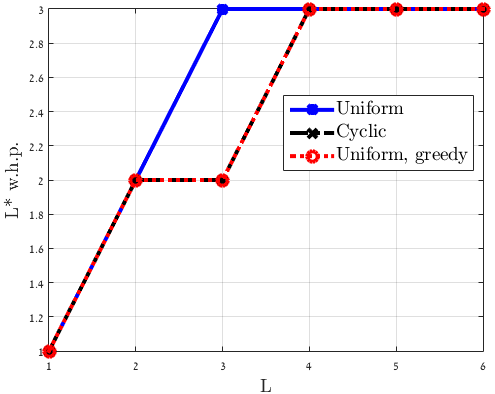}
                \caption{$k=3, n=6.$}
                \label{k3n6_whp}
        \end{subfigure}

        \caption{The expected $L^*$ with probability at least $0.95$ as a function of the load $L$ ($N=12$). \label{whp_results}
}
\end{figure*}

In Fig. \ref{fig:full_tp_comp}, we compare the probability of a full-throughput solution (i.e., $L^*=L$ solution) for the uniform, cyclic and design placement policies for $k=3, n=5$ and $L=3$. We note that in the design case, the probability is obtained analytically using the balls and bins model (See Section \ref{subsec:design_prob}), given the number of valid MU sets (blocks). To find the number of blocks, we used constant-weight codes with $n=5$ and $d=2\left( {n - {t_{\max }}} \right) = 6$ (see Section \ref{subsec:design_placement}). The (exact) number of blocks (i.e., the maximum number of codewords) in this case is known for $N$ values up to $17$ \cite{Agrell}. The graphs in Fig. \ref{fig:full_tp_comp} show that the uniform placement is somewhat better than the design placement in terms of full-throughput probability. The reason is the large number of valid MU sets in the uniform case, which is ${N \choose n}$, compared to the number of blocks in the design case, which is typically much smaller (e.g., $68$ blocks when $N=17$ compared to $6188$ subsets). However, as no efficient read algorithm is known for the uniform placement, an exhaustive-search algorithm requires ${2^{Ln}} = {2^{15}}$ operations to find an optimal solution. On the other hand, the efficient optimal read algorithm in the design case requires only $L\left( {n + L} \right) = 24$ operations, i.e., a number smaller by four orders of magnitude. This makes the design placement policy appealing in practice due to complexity considerations.

\begin{figure}
\centering
\includegraphics[scale=0.62]{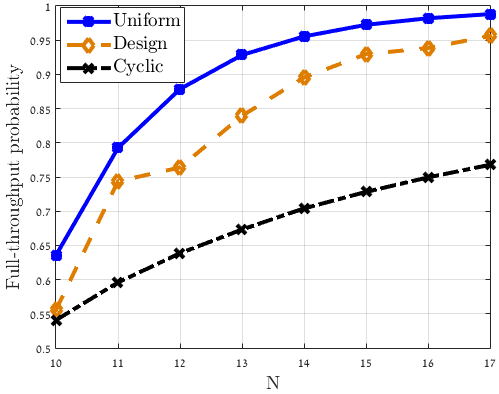}
\caption{Full-throughput probability for $k=3, n=5, L=3$ as a function of $N$.}
\label{fig:full_tp_comp}
\end{figure}

\section{Conclusion}
\label{sec:conclusions}

In this paper, we studied placement policies and read algorithms of coded packets in a switch memory. The study revealed that coding can significantly improve switching throughput, and that the choice of placement has significant effect on performance and complexity. We proved that in its most general form, the problem of obtaining maximum throughput for a set of requested packets is a hard problem. Therefore, we moved to propose two practical placement policies and efficient optimal read algorithms, with better throughput performance in certain cases compared to the general non-structured placement policy.

We demonstrated tradeoffs between write flexibility, read-algorithm complexity and performance. In particular, we saw that no choice of placement policy is universally optimal, and provided analytic tools for choosing a policy wisely. Our work leaves many interesting problems for future research. For example, one may consider other structured write policies by imposing different constraints rather than restricting pairwise intersections. It is also interesting to consider variable-length codes (i.e., varying values of $k$ and $n$ values for each packet), to match the expected switch load.

\bibliographystyle{IEEEtran}

\bibliography{MDS_paper}

\appendices

\section{Detailed proof of Theorem \ref{np_hard}}
\label{np_hard_detailed_proof}
To show the hardness of nkMTP when $3 \le k \le n$, we define its decision-problem version, which we name $M$-nkMTP. In the rest of this appendix, we assume that $3 \le k \le n$.

\begin{problem}{($M$-nkMTP)}
\label{nkMTP_decision}

\textbf{Input}: An nkMTP instance and a positive integer $M$.

\textbf{Output}: "Yes" if there are $M$ subsets $\mathcal{S}'_i \subseteq \mathcal{S}_i$ with the properties $\left| \mathcal{S}'_i \right| = k$, $\mathcal{S}'_i \cap \mathcal{S}'_j = \emptyset$ ($i \ne j$).
\end{problem}
For showing that nkMTP is NP-hard we can equivalently show that $M$-nkMTP is NP-complete. Note that $M$-nkMTP is in NP, since once we are given a collection of $M$ subsets $S'_i \subseteq \mathcal{S}_i$ claimed to be pairwise disjoint, this can be validated in polynomial time. It remains to reduce a known NP-complete problem to $M$-nkMTP, meaning that we have to show that an efficient solution to $M$-nkMTP implies an efficient solution to this NP-complete problem. We will reduce the \textit{$l$-set packing} problem ($l$-SP), known to be NP-complete for $l \ge 3$  \cite{Hazan}, to our problem. $l$-SP is defined as follows.
\begin{problem}{($l$-SP)}
\label{l_SP}

\textbf{Input}: A collection of sets over a certain domain, each of them of size $l$, and a positive integer $M$.

\textbf{Output}: "Yes" if there are $M$ pairwise disjoint sets.
\end{problem}
$M$-nkMTP is NP-complete for $3 \le k = n$, since in this case $M$-nkMTP and $l$-SP, for $l=k=n$, are essentially the same. Therefore, it remains to reduce $l$-SP ($l \ge 3$) to $M$-nkMTP for $3 \le k < n$. Let us begin with reducing $l$-SP to $M$-nkMTP with $k=l, n = k+1$.

Consider an instance of $l$-SP with $l=k$, with $M$ denoting the number of pairwise disjoint subsets required in the solution. Assume that the input to $l$-SP are $L$ sets $\mathcal{A}_i$ ($i=1,2,...,L$), where the elements contained in $\mathcal{A}_i$ are $\bigcup\limits_i {{\mathcal{A}_i}}  = \left\{ {{a_1},{a_2},...,{a_s}} \right\}$. For building an instance of $M$-nkMTP with $k=l, n=k+1$, do the following:

\begin{itemize}
\item Build sets $\mathcal{B}_i$, each of size $k$, from $s$ new elements $\left\{ {{b_1},{b_2},...,{b_s}} \right\}$, such that a one-to-one correspondence between the elements in $\mathcal{A}_i$ and the elements in $\mathcal{B}_i$ exists: ${a_j} \in {\mathcal{A}_i} \Leftrightarrow {b_j} \in {\mathcal{B}_i}$.
\item Add a new element, say $\theta$, which does not belong to either $\mathcal{A}_i$ or $\mathcal{B}_i$, to both sets to obtain the new sets denoted by $\tilde{\mathcal A}_i$ and $\tilde{\mathcal B}_i$.
\end{itemize}
The input to $M$-nkMTP with $n=k+1$ will be the sets $\tilde{\mathcal A}_i$ and $\tilde {\mathcal B}_i$, where we ask whether there exist $2M$ subsets of size $k$ each that are pairwise disjoint. If $l$-SP provides a solution of size $M$ for the sets $\mathcal{A}_i$, then clearly the sets ${\mathcal{A}_i} \subseteq \tilde {\mathcal A}_i, {\mathcal{B}_i} \subseteq \tilde {\mathcal B}_i$ serve as solution of size $2M$ to $M$-nkMTP with $n=k+1$. On the other hand, if there exists a solution of size $2M$ in the $M$-nkMTP problem, we have three cases:
\begin{enumerate}
\item $M$ subsets $\mathcal{A}{'_i} \subseteq \tilde {\mathcal A}_i$ and $M$ subsets $\mathcal{B}{'_i} \subseteq \tilde{\mathcal B}_i$ appear in the solution. The element $\theta$ can appear in only one of the subsets, since they must be pairwise disjoint. If $\theta$ belongs to some $\mathcal{A}{'_i}$, then we have $M$ subsets ${\mathcal{B}}{'_i}$ that provide a solution to $l$-SP (after transforming the elements in $\mathcal{B}{'_i}$ to the their corresponding elements in ${\mathcal{A}}{'_i}$). On the other hand, if $\theta$ belongs to some $\mathcal{B}{'_i}$, then the solution is the sets $\mathcal{A}{'_i}$.
\item $M_1$ subsets $\mathcal{A}{'_i} \subseteq \tilde {\mathcal A}_i$ and $M_2$ subsets $\mathcal{B}{'_i} \subseteq \tilde {\mathcal B}_i$ appear in the solution, where $M_1 < M_2$ and $M_1 + M_2 = 2M$. $\theta$ can appear in at most one of the subsets $\mathcal{B}'_i$. In addition, $M < M_2$, and therefore choosing the subsets $\mathcal{B}'_i$ that do not contain $\theta$ leads to a solution of $l$-SP with at least $M$ subsets (again, transformation to the elements of $\mathcal{A}_i$ is required).
\item $M_1$ subsets $\mathcal{B}{'_i} \subseteq \tilde {\mathcal B}_i$ and $M_2$ subsets $\mathcal{A}{'_i} \subseteq \tilde {\mathcal A}_i$ appear in the solution, where $M_1 < M_2$ and $M_1 + M_2 = 2M$. A solution of size at least $M$ to $l$-SP is obtained in a similar way to the previous case.
\end{enumerate}

The transformation $\mathcal{A}_i, \mathcal{B}_i \to \tilde{\mathcal A}_i,\tilde {\mathcal B}_i$ is polynomial in $L$, since it merely requires to build $L$ sets of size $k$ and to add one element to each of the resulting $2L$ sets. Thus, the reduction described above is a polynomial time reduction. Therefore, $M$-nkMTP is NP-complete for $k \ge 3, n = k+1$, and it remains to show that $M$-nkMTP is NP-complete for $k \ge 3, n > k+1$. Consider $M$-nkMTP with $k \ge 3, n = k+2$. We can reduce $M$-nkMTP with $k \ge 3, n = k+1$ (which we proved to be NP-complete) to $M$-nkMTP with $k \ge 3, n = k+2$, similarly to the reduction of $l$-SP to $M$-nkMTP with $k=l, n=k+1$ that was described earlier. Continuing in the same fashion, we are able to reduce $M$-nkMTP with $n=k+j$ ($k \ge 3, j \ge 1$) to $M$-nkMTP with $n=k+j+1$. Finally, we deduce that $M$-nkMTP is NP-complete for $3 \le k \le n$, meaning that nkMTP (the optimization version of $M$-nkMTP) is NP-hard. \qed

\section{$p_{\rm cover}^{\rm uni}$ and $p_{\rm cover}^{\rm cyc}$}
\label{appendix:p_cover}

The following derivation of $p_{\rm cover}^{\rm uni}$ (see Section \ref{subsec:uniform_prob}) is based on the union model \cite{QPEC}. Define the function:
\begin{equation}
{I_m}\left( {i,n} \right)\\\nonumber = \sum\limits_{j = 0}^{\min(i,n)-m} {{{\left( { - 1} \right)}^j} \cdot \nu_{m+j}\left( {i,n} \right)\cdot{m+j \choose m} },
\end{equation}
where
\begin{equation}
\nu_{m+j}\left( {i,n} \right) = {N \choose m+j} \cdot {N-(m+j) \choose i-(m+j)} \cdot {N-(m+j) \choose n-(m+j)},
\end{equation}
for $i=0,1,...,N$. $I_m$ is the number of ways to realize two sets of cardinalities $i$ and $n$, taken from a set of $N$ elements, such that their intersection is of cardinality $m$. Define the probability distribution $R_m$:
\begin{equation}
{R_m}\left( {i,n} \right) = \frac{{{I_m}\left( {i,n} \right)}}{{N \choose i} \cdot {N \choose n}},
\end{equation}
which is the probability that two sets of cardinalities $i$ and $n$, taken uniformly at random from a set of $N$ elements, have an intersection of cardinality $m$. Define the following $\left( {N + 1} \right) \times \left( {N + 1} \right)$ Markov matrix, with indices $i,j$ ranging from $0$ to $N$:
\begin{equation}
{\left( {\bf \Gamma}  \right)_{i,j}} = {R_{i + n - j}}\left( {i,n } \right).
\end{equation}
The $(i,j)$ entry of ${\bf{\Gamma}}$ is the probability that the union of a set with $i$ elements and a set with $n$ elements is of cardinality $j$. Finally,  $p_{\rm cover}^{\rm uni}$ is the sum of the first $kL$ entries in the first row of ${{\bf{\Gamma}}^L}$  (we assume that $kL \le N$).

To obtain $p_{\rm cover}^{\rm cyc}$, we use the probability distribution on $V$, the number of vacant points on the circle when $L$ random arcs are placed without replacement. A closed-form expression for this distribution is given in Theorem 1 in \cite{Barlevy}. This expression is rather long and depends on the parameter range so we do not provide this here. We are actually interested in the number of \textit{non}-vacant points (i.e., how many MUs are covered by the packets) which is the probability distribution of $N-V$. As in our case the arcs are taken \textit{with} replacement, we condition the probability distribution of $N-V$ by the probability distribution on the number of distinct arcs among $L$ random arcs using the balls-and-bins model (i.e., \eqref{balls_and_bins} with $b=L$). We note here that sampling with replacement is discussed as well in Chapter 4.1 of \cite{Barlevy}.

\end{document}